\documentclass[english]{article}
\usepackage[latin9]{inputenc}
\usepackage{geometry}
\usepackage{float}
\usepackage{mathrsfs}
\usepackage{amsmath}
\usepackage{amssymb}
\usepackage{graphicx}
\usepackage{enumerate}

\makeatletter

\floatstyle{ruled}
\newfloat{algorithm}{tbp}{loa}
\providecommand{\algorithmname}{Algorithm}
\floatname{algorithm}{\protect\algorithmname}

\@ifundefined{date}{}{\date{}}

\usepackage{float}\usepackage{mathrsfs}\usepackage{amsfonts}\usepackage{bm}\usepackage{subfigure}\usepackage{amsthm}
\usepackage{epsfig}

\usepackage[auth-sc,affil-it]{authblk}

\@ifundefined{definecolor}{\@ifundefined{definecolor}
 {\@ifundefined{definecolor}
 {\usepackage{color}}{}
}{}
}{}

\usepackage[all]{xy}

\newtheorem{theorem}{Theorem}[section]
\newtheorem{lemma}{Lemma}[section]

\newtheorem{proposition}{Proposition}[section]
\newtheorem{cor}{Corollary}[section]
\newtheorem{ass}{Assumption}[section]
\newcounter{hypA}
\newenvironment{hypA}{\refstepcounter{hypA}\begin{itemize}
  \item[({\bf A\arabic{hypA}})]}{\end{itemize}}
\newcounter{hypB}

\usepackage{babel}

\usepackage{babel}

\newread\mysource
\begin{document}

\title{Bayesian Static Parameter Estimation for Partially
Observed Diffusions via Multilevel Monte Carlo}
\date{}
\author[1]{Ajay Jasra\thanks{staja@nus.edu.sg}}
\author[2]{Kengo Kamatani\thanks{kamatani@sigmath.es.osaka-u.ac.jp}}
\author[3]{Kody J.~H.~Law\thanks{kodylaw@gmail.com}}
\author[4]{Yan Zhou\thanks{stazhou@nus.edu.sg}}
\affil[1,4]{Department of Statistics  Applied Probability, National University of Singapore, SG}
\affil[2]{Department of Engineering Science, Osaka University, JP}
\affil[3]{Computer Science and Mathematics Division, Oak Ridge National Laboratory, USA}

\maketitle

\begin{abstract}
In this article we consider static Bayesian parameter estimation for partially observed diffusions that are
discretely observed. We work under the assumption that one must resort to discretizing the underlying diffusion
process, for instance using the Euler Maruyama method. 
Given this assumption, we show how one can use Markov chain
Monte Carlo (MCMC) and particularly particle MCMC 
[Andrieu, C., Doucet, A. \& Holenstein, R. (2010). Particle Markov chain Monte Carlo methods (with discussion). {\it J. R. Statist. Soc. Ser. B}, 72, 269--342]
to implement a new approximation of the multilevel (ML) Monte Carlo (MC)
collapsing sum identity. Our approach comprises constructing an approximate coupling of the posterior density
of the joint distribution over parameter and hidden variables
at two different discretization levels and then correcting 
by an importance sampling method. 
The variance of the weights are independent of the length of the observed data set.
The utility of such a method is that, 
for a prescribed level of mean square error, the cost of this MLMC method is provably less than 
i.i.d.~sampling from the posterior associated to the most precise discretization. 
However the method here comprises using only known and efficient simulation methodologies.
The theoretical results are illustrated by inference of the parameters of
two prototypical processes given noisy partial observations of the process:
the first is an Ornstein Uhlenbeck process and the second is a more general Langevin equation.
\\

\end{abstract}

\section{Introduction}

The Hidden Markov Model (HMM) is widely used in many disciplines, including applied mathematic, statistics, economics and finance; see \cite{cappe} for an overview.
In this article, we are interested in HMMs 
given by diffusions which are partially observed, 
discretely in time. 
In particular, we assume that in order to fit the model
to the data, one \emph{must} resort to a discretization of the diffusion, for instance, 
using Euler-Maruyama. In addition, we assume that associated to the model is a static 
(non-time-varying)  finite dimensional parameter, 
which one is interested to infer given a fixed length data record. In simple terms, the discretization, of level $h$ say, where
as $h\rightarrow 0$ one obtains the exact diffusion, induces a posterior say $\pi_h$ on the static parameter $\theta$ and hidden states at the observation times, say $X_{0:n}$. We seek to approximate
$\mathbb{E}_{\pi_h}[\varphi(\theta,X_{0:n})]$ for appropriately defined real-valued functions. 
Ultimately, one might seek to remove the dependence upon $h$ and get the
exact expectation with no discretization bias. We remark that the model will be formally introduced in the next section.
This framework is relevant to a broad range of applications in science and engineering; see \cite{cappe,law2015data}


The task of computing the expectation for any fixed $h>0$ is a non-trivial task, which often requires quite advanced Monte Carlo methods. As has been remarked in many articles
in the literature, ofen the joint correlation between $\theta$ and $X_{0:n}$ means even standard MCMC methods may produce very inaccurate of inefficient approximations
of the expectation of interest, despite their theoretical validity. An important algorithm that has, to an extent, helped to alleviate these difficulties is the particle MCMC (PMCMC)
methods of \cite{andrieu} and their subsequent developments (e.g.~\cite{corr_pm}). Intrinsically, this method uses a sequential Monte Carlo (SMC) (e.g.~\cite{doucet_johan}) method to help move the
samples around the state-space, for instance, inside a Metropolis-Hastings acceptance/rejection scheme, although Gibbs versions also exist. 
PMCMC delivers a Markov chain which provides consistent estimates of expectations of the form $\mathbb{E}_{\pi_h}[\varphi(\theta,X_{0:n})]$, for any fixed $h$
SMC methods are well-known as being efficient techniques for filtering, when the state-variable at time $k$, $X_k$, is of moderate to low dimension and  all the static parameters are fixed.

In the context of this article, there is an additional degree of freedom, which can be utilized to further enhance the PMCMC method. This is associated to the discretization level $h$.
We consider using the multilevel Monte Carlo (MLMC) framework  
\cite{giles:08, giles_acta,heinrich2001multilevel}.
This allows one to leverage in an optimal way 
the nested problems arising in this context, hence minimizing the necessary cost to obtain a given level of mean square error.  
Set $\pi$ as the posterior on $\theta,X_{0:n}$ with no discretization bias
 and $\pi_{h_l}$ as the time-discretized posterior on $\theta,X_{0:n}$ with time discretization $h_l$, one
has for an intergrable, real-valued function $\varphi$ and $+\infty>h_0>h_1>\cdots>h_L>0$ (the levels)
\begin{equation}\label{eq:telescope}
\mathbb{E}_{\pi_{h_L}}[\varphi(\theta,X_{0:n})] = \sum_{l=0}^L\{\mathbb{E}_{\pi_{h_l}}[\varphi(\theta,X_{0:n})] - \mathbb{E}_{\pi_{h_{l-1}}}[\varphi(\theta,X_{0:n})]\}
\end{equation}
where $\mathbb{E}$ is the expectation operator and $\mathbb{E}_{\pi_{h_{-1}}}[\varphi(\theta,X_{0:n})]:=0$.
The idea of MLMC is then to approximate each summand by independently simulating $N_l$ samples from a dependent coupling of $(\pi_{h_l},\pi_{h_{l-1}})$.
In such scenarios, one can show that the overall mean square error (MSE) associated to the approximation of $\mathbb{E}_{\pi}[\varphi(\theta,X_{0:n})]$ is:
\begin{equation}\label{eq:mse}
\textrm{MSE} = \textrm{Bias}(L,\varphi)^2 + \sum_{l=0}^L \frac{V_l}{N_l} \ ,
\end{equation}
where 
\begin{equation}\label{eq:bias}
\textrm{Bias}(L,\varphi) = |\mathbb{E}_{\pi_{h_L}}[\varphi(\theta,X_{0:n})]-\mathbb{E}_{\pi} [\varphi(\theta,X_{0:n})]| \ ,
\end{equation} 
and $0<V_l<+\infty$ are a collection of constants.
It is remarked that it is the coupled samples which induce $V_l$ to be a function of $h_l$ which is often critical as we explain below.
Assuming the cost of $C_l$ per level, per sample, the cost of the algorithm is then $\sum_{l=0}^L C_l N_l$.
Fixing $\epsilon>0$ and given an appropriate parameterization of $h_l$ (e.g.~$h_l=2^{-l}$), one then chooses $L$ to ensure that $\textrm{Bias}(L,\varphi)^2=\mathcal{O}(\epsilon^2)$ and then given $C_l,V_l$ characterised
as a function of $h_l$ optimizes $N_0,\dots,N_L$ to minimize the cost so that the term $\sum_{l=0}^L \frac{V_l}{N_l}=\mathcal{O}(\epsilon^2)$;
\cite{giles:08} gives the solution to this constrained optimization problem. In many scenarios of practical interest the associated MLMC algorithm can achieve a MSE of $\mathcal{O}(\epsilon^2)$ at a cost which is less than i.i.d.~sampling from
$\pi_{h_L}$; note that 
this has not yet been established in the problem under study here. 
The main issue is that sampling independently from the couples $(\pi_{h_l},\pi_{h_{l-1}})$ 
is not possible in our context.

In this paper we show how to implement a new approximation of the multilevel
collapsing sum identity. Our approach comprises constructing an approximate coupling of the posterior density
of the joint on the parameter and hidden space at two different discretization levels and then correcting 
by an importance sampling method, whose variance of the weights are independent of the length of the observed data set.
The utility of such a method is that it comprises using known and efficient simulation methodologies, 
instead of coupling
algorithms as explored in \cite{jacob,mlpf,ourmlpf1,sen}. 
In particular, our approach facilitates a mathematical analysis which allows us to establish that 
our approach can be better than sampling (e.g.~by PMCMC) from the posterior associated to the most precise discretization.
The algorithm presented here is distinct from either of the previously introduced multilevel
MCMC (MLMCMC) algorithms \cite{hoan:12, ketelsen2013hierarchical}, and may be generalized. 

This article is structured as follows. In Section \ref{sec:model} the model is described. In Section \ref{sec:app_theory} we describe our approach and
give a mathematical result associated to the MSE of the method. In Section \ref{sec:numerics} we give practical simulations to establish the theory.
The appendix contains some of the proofs for the result of Section \ref{sec:app_theory}.

\section{Model}\label{sec:model}

We consider the following partially-observed diffusion process:
\begin{eqnarray}
dX_t & = & a_{\theta}(X_t)dt + b_{\theta}(X_t)dW_t
\label{eq:sde}
\end{eqnarray}
with $X_t\in\mathbb{R}^d=\mathsf{X}$, $t\geq 0$, $X_0$ has initial probability density $f_{\theta}$ and $\{W_t\}_{t\in[0,T]}$ a Brownian motion of appropriate dimension. 
$\theta\in\Theta\subseteq\mathbb{R}^{d_\theta}$ is a static parameter of interest.
The following assumptions will be made on the diffusion process. 
\begin{ass} 
$a_{\theta}:\mathbb{R}^d\rightarrow\mathbb{R}^d$, $b_{\theta}:\mathbb{R}^d\rightarrow\mathbb{R}^{d\times d}$  satisfy 
\begin{itemize}
\item[{\rm (i)}] {\bf global Lipschitz property}:
there is a $C>0$ such that 
$|a_{\theta}(x)-a_{\theta}(y)|+|b_{\theta}(x)-b_{\theta}(y)| \leq C |x-y|$ 
for all $x,y \in \mathsf{X}$  
and all $\theta\in\Theta$; 
\item[{\rm (ii)}] {\bf bounded moments}: $\sup_{\theta\in\Theta}\mathbb{E}_\theta |X_0|^p < \infty$ for all $p \geq 1.$
\end{itemize}
\label{ass:diff}
\end{ass}
Notice that (i) and (ii) together imply that $\mathbb{E}_\theta |X_n|^p < \infty$ for all $n$.

It will be assumed that the data are 
regularly spaced (i.e.~in discrete time) observations 
$y_1,\dots,y_{n}$, $y_k \in \mathbb{R}^m=\mathsf{Y}$.
It is assumed that conditional on $X_{k\delta}$, for discretization $\delta>0$,
$Y_k$ is independent of all other random variables with density $g_{\theta}(x_{k\delta},y_k)$.
For simplicity of notation let $\delta=1$ (which can always be done by rescaling time), so $X_k = X_{k\delta}$.
It is noted that we assume that one does not have access to a non-negative and unbiased estimate
of the transition density of the diffusion and we are forced to work with a discretized process.

The above formulation can then summarized as follows, on discretizing the diffusion process with discretization level $h$.
We have a pair of discrete-time stochastic
processes, $\left\{  X_{n}\right\}  _{n\mathbb{\geq}0}$ and $\left\{
Y_{n}\right\}  _{n\geq 1}$, where $X_{n}\in\mathsf{X}$ (with associated $\sigma-$algebra $\mathcal{X}$) is an unobserved
process and $y_{n}\in\mathsf{Y}$ (with associated $\sigma-$algebra $\mathcal{Y}$) is observed. Let $\theta\in\Theta\subseteq\mathbb{R}^{d_\theta}$ be a parameter .
The hidden process
$\left\{  X_{n}\right\}  $ is a\ Markov chain with initial density $f_{\theta} $ at time $0$ and transition density $f_{\theta,h}\left(x_{p-1},x_{p}\right)
$, i.e.~for each $\theta\in\Theta$
\begin{equation}
\mathbb{P}_{\theta,h}(X_{0}\in A)=\int_A f_{\theta}(x)dx \quad\text{ and }%
\quad\mathbb{P}_{\theta,h}(X_{p}\in A|X_{p-1}=x_{p-1})=\int_{A}f_{\theta,h}(x_{p-1},x_p
)dx_{p}\quad p\geq1 \label{eq:evol}%
\end{equation}
where $\mathbb{P}_{\theta,h}$ denotes probability, $A\in\mathcal{X}$ and
$dx_{n}$ is a dominating $\sigma$-finite measure.
In addition, the observations
$\left\{  Y_{n}\right\}  _{n\geq 1}$\ conditioned upon $\left\{  X_{n}\right\}
_{n\mathbb{\geq}0}$ are statistically independent and have marginal density
$g_{\theta}\left(x_{n},y_{n}\right)  $, i.e.%
\begin{equation}
\mathbb{P}_{\theta,h}(Y_{n}\in B|\{X_{k}\}_{k\geq 0}=\{x_{k}\}_{k\geq 0})=\int_{B}%
g_{\theta}(x_{n},y_{n})dy_{n}\quad n\geq 1  \label{eq:obs}%
\end{equation}
with $B\in\mathcal{Y}$ and $dy_{n}$ the dominating $\sigma$-finite measure. The HMM is given by equations (\ref{eq:evol})-(\ref{eq:obs}) and is often referred to in the literature as a state-space model. In our context  $\theta\in\Theta$ is a parameter of interest with prior $\pi_{\theta}$.

Given the joint density on $\mathsf{U}:=\Theta\times\mathsf{X}^{n+1}$ 
$$
\pi_{h}(\theta,x_{0:n}) \propto \pi_{\theta}(\theta)f_\theta(x_0) 
\prod_{p=1}^n g_{\theta}(x_p,y_p) f_{\theta,h}(x_{p-1},x_p) \ ,
$$
for $\varphi\in\mathcal{B}_{b}(\mathsf{U})\cap\textrm{Lip}(\mathsf{U})$, 
where $\mathcal{B}_{b}(\mathsf{U})$ are the bounded and real-valued measurable functions
on $\mathsf{U}$ and $\textrm{Lip}(\mathsf{U})$ are the Lipschitz, measurable functions
on $\mathsf{U}$, and for $+\infty>h_0>\cdots>h_L>0$ we would like to compute 
\begin{equation}\label{eq:ml_id}
\mathbb{E}_{\pi_{h_L}}[\varphi(\theta,X_{0:n})] = \sum_{l=0}^L\Big\{\mathbb{E}_{\pi_{h_l}}[\varphi(\theta,X_{0:n})] - \mathbb{E}_{\pi_{h_{l-1}}}[\varphi(\theta,X_{0:n})]
\Big\}
\end{equation}
where $\mathbb{E}_{\pi_{h_{-1}}}[\cdot] = 0$. We will use the MLMC approach.

Consider only a single pair $\mathbb{E}_{\pi_{h}}[\varphi(\theta,X_{0:n})] - \mathbb{E}_{\pi_{h'}}[\varphi(\theta,X_{0:n})]$, $h<h'$.
It is well known that if one can sample from a dependent coupling of $(\pi_h,\pi_{h'})$, such as the maximal coupling, then Monte Carlo estimation of
such a difference can be performed at a lower cost than i.i.d~sampling from the independent coupling of $(\pi_h,\pi_{h'})$ \cite{giles:08,giles_acta}.
The main issue is that such couplings are typically not available up-to a non-negative and unbiased estimator.
We consider the scenario where one samples from a sensible, approximate, coupling and corrects via importance sampling.

\section{Method and Analysis}\label{sec:app_theory}

\subsection{Method}

We are to approximate the identity \eqref{eq:ml_id}.
Our procedure, when considering the summands from $1,\dots,L$ will be to run $L$ independent pairs of the idea to be described below. The case $l=0$ is simply using (e.g.) PMCMC to approximate $\mathbb{E}_{\pi_{h_0}}[\varphi(\theta,X_{0:n})]$; we refer
the reader to \cite{andrieu} for details on PMCMC - a simple decsription is below.
We only consider a pair $\mathbb{E}_{\pi_{h}}[\varphi(\theta,X_{0:n})] - \mathbb{E}_{\pi_{h'}}[\varphi(\theta,X_{0:n})]$, $h<h'$.
The methodology and analysis in this context of one pair will suffice to justify our approach as we will explain below.

Let $z=(x,x')\in\mathsf{X}\times\mathsf{X}=\mathsf{Z}$ and $Q_{\theta,h,h'}(z,\bar{z})$ be any coupling (other than the independent one) of $(f_{\theta,h}(x,\bar{x}),
f_{\theta,h'}(x',\bar{x}'))$. For instance, in the context of an Euler discretization a description can be found in \cite{ourmlpf1} (see also appendix \ref{app:couple_euler}). Let $G_{p,\theta}(z) = \max\{g_{\theta}(x,y_p),g_{\theta}(x',y_p)\}$ (note that alternative choices of $G_{p,\theta}$ are possible).
We propose to sample from the probability density on $\mathsf{V} = \Theta\times\mathsf{X}^{2n+2}$ (write the associated $\sigma-$algebra as $\mathcal{V}$)
$$
\pi_{h,h'}(\theta,z_{0:n}) \propto \pi_{\theta}(\theta)\nu_\theta(z_0)
\prod_{p=1}^n G_{p,\theta}(z_p) Q_{\theta,h,h'}(z_{p-1},z_p).
$$
Then for $\varphi\in\mathcal{B}_{b}(\mathsf{U})\cap\textrm{Lip}(\mathsf{U})$:
$$
\mathbb{E}_{\pi_h}[\varphi(\theta,X_{0:n})] - \mathbb{E}_{\pi_{h'}}[\varphi(\theta,X_{0:n})] =
$$
\begin{equation}
\frac{\mathbb{E}_{\pi_{h,h'}}[\varphi(\theta,X_{0:n})H_{1,\theta}(\theta,Z_{0:n})]}{\mathbb{E}_{\pi_{h,h'}}[H_{1,\theta}(\theta,Z_{0:n})]} - 
\frac{\mathbb{E}_{\pi_{h,h'}}[\varphi(\theta,X_{0:n}')H_{2,\theta}(\theta,Z_{0:n})]}{\mathbb{E}_{\pi_{h,h'}}[H_{2,\theta}(\theta,Z_{0:n})]}
\label{eq:master_eq}
\end{equation}
where
\begin{eqnarray*}
H_{1,\theta}(\theta,z_{0:n}) & = & \prod_{p=1}^n\frac{g_{\theta}(x_p,y_p)}{G_{p,\theta}(z_p)} \\
H_{2,\theta}(\theta,z_{0:n}) & = & \prod_{p=1}^n\frac{g_{\theta}(x_p',y_p)}{G_{p,\theta}(z_p)}.
\end{eqnarray*}
We note that our choice of $G_{p,\theta}(z)$ ensures that $H_{1,\theta}$ and $H_{2,\theta}$ are uniformly
upper-bounded by 1 and hence that the variance w.r.t.~any probability is independent of $n$.

\subsubsection{Particle MCMC}

Let $(\mathsf{W},\mathcal{W})$ be a measurable space such that $\mathsf{V}\subseteq \mathsf{W}$.
Let $K:\mathsf{W}\times\mathcal{W}\rightarrow[0,1]$ be any ergodic Markov kernel of invariant measure
$\eta$ such that one can consistently estimate expectations w.r.t.~$\pi_{h,h'}$. 
For instance, if
for every $A\in\mathcal{V}$
$$
\int_{A\times(\mathsf{W}\setminus\mathsf{V})} \eta(dw) = \int_{A} \pi_{h,h'}(\theta,z_{0:n}) d(\theta,z_{0:n}).
$$
Our construction allows a particle MCMC approach to be adopted, 
which is not quite as the displayed equation, but nonetheless allows
one to infer $\pi_{h,h'}$. We focus on one particle MCMC method for completeness, but, we reiterate that one can use
the analysis here for more advanced versions of the algorithm, or indeed, any MCMC of the form above.

We will now describe the particle marginal Metropolis-Hastings (PMMH) algorithm. 
Let $M\geq 1$ and $\theta$ be fixed, and introduce random variables $a_{0:n-1}\in\{1,\dots,M\}^n$,
which will denote the indices of the selected particles upon resampling at the given steps. 
One can run a particle filter \cite{delm:04} 
to approximate 
$$
\pi_{h,h'}(z_{0:n}|\theta) \propto \nu_\theta(z_0)\prod_{p=1}^n G_{p,\theta}(z_p) Q_{\theta,h,h'}(z_{p-1},z_p)
$$ 
by sampling from the following joint, on the
space $\{1,\dots,M\}^{n}\times\mathsf{Z}^{M(n+1)}$
\begin{equation}\label{eq:pf_law}
p(a_{0:n-1}^{1:M},z_{0:n}^{1:M}|\theta) = 
\Big(\prod_{i=1}^M  \nu_{\theta}(z_0^i)\Big)\prod_{p=1}^n \prod_{i=1}^M\Big( \frac{G_{p-1,\theta}(z_{p-1}^{a_{p-1}^i})}{\sum_{j=1}^M G_{p-1,\theta}(z_{p-1}^j)} Q_{\theta,h,h'}(z_{p-1}^{a_{p-1}^i},z_p^i)\Big) \ ,
\end{equation}
where $G_{0,\theta}:=1$.
Note that better algorithms can be constructed, but we just present the most simple approach. We remark that
\begin{equation}\label{eq:pf_nc}
p^M_{h,h'}(y_{0:n}|\theta) = 
\prod_{p=1}^{n}\Big(\frac{1}{M}\sum_{j=1}^M G_{p,\theta}(z_{p}^j)\Big)
\end{equation}
is an unbiased estimator of $p_{h,h'}(y_{0:n}|\theta) = \int_{\mathsf{Z}^{n+1}} \nu_\theta(z_0)\prod_{p=1}^n G_{p,\theta}(z_p) Q_{\theta,h,h'}(z_{p-1},z_p) dz_{0:n}$; see \cite{delm:04}.

The PMMH algorithm works as follows. 
The superscripts for $(\theta,k)$ are the iteration (time) counter of the MCMC.
\begin{enumerate}
\item{Initialize: Sample $\theta^{0}$ from the prior and then sample 
$(a_{0:n-1}^{1:M},z_{0:n}^{1:M})$
from $p(a_{0:n-1}^{1:M},z_{0:n}^{1:M}|\theta^{0})$ as in \eqref{eq:pf_law}, 
and store $p^M_{h,h'}(y_{0:n}|\theta^{0})$ as in \eqref{eq:pf_nc}.
Select a path $z_{0:n}^{j}$, constructed by drawing $z_n^j$ with probability proportional to 
$G_{n,\theta^{0}}(z_{n}^j)$, and setting $(z^{j'}_{p-1}|z^{j'}_p)=z_{p-1}^{a_{p-1}^{j'}}$;} 
set $k^{0}$ as the index of the selected path. Set $i=1$.
\item{Iterate: Sample $\theta'|\theta^{i-1}$ according to a proposal with conditional density $q(\theta'|\theta^{i-1})$  then from $p(a_{0:n-1}^{1:M},z_{0:n}^{1:M}|\theta')$ as in \eqref{eq:pf_law}.
Select a path $z_{0:n}^{j}$ with probability proportional to $G_{n,\theta'}(z_{n}^j)$ 
and constructed as described above; set $k'$ as the index of the selected path. Set $\theta^{i}=\theta'$, $k^{i}=k'$ with probability:
$$
1\wedge \frac{p^M_{h,h'}(y_{0:n}|\theta')}{p^M_{h,h'}(y_{0:n}|\theta^{i-1})}\frac{\pi_{\theta}(\theta')q(\theta^{i-1}|\theta')}{\pi_{\theta}(\theta^{i-1})q(\theta'|\theta^{i-1})}
$$
otherwise $\theta^{i}=\theta^{i-1}$, $k^{i}=k^{i-1}$. Set $i=i+1$ and return to the start of 2.}
\end{enumerate}
We denote by $K$ the PMMH kernel and denote by $(\mathsf{W},\mathcal{W})$ the measurable space for which it is defined upon. The invariant measure is denoted $\eta$. For the analysis, we assume the MCMC algorithm
is started in stationarity.

Then one estimates \eqref{eq:master_eq} by
$$
\frac{\frac{1}{N}\sum_{i=1}^N \varphi(\theta^i,x_{0:n}^{k^i})H_{1,\theta^i}(\theta^i,z_{0:n}^{k^i})}{\frac{1}{N}\sum_{i=1}^N H_{1,\theta^i}(\theta^i,z_{0:n}^{k^i})} 
- \frac{\frac{1}{N}\sum_{i=1}^N\varphi(\theta^i,x_{0:n}'^{k^i})H_{2,\theta^i}(\theta^i,z_{0:n}^{k^i})}{\frac{1}{N}\sum_{i=1}^NH_{2,\theta^i}(\theta^i,z_{0:n}^{k^i})}.
$$
This estimate is consistent in the limit as $N$ grows; see \cite{andrieu}. 
To simplify the notation we replace 
$k^i$ in the superscripts by $i$ from here on.

\subsection{Multilevel Considerations}\label{sec:mlmc}

As described for MLMC in the introduction, we will approximate the expectation 
using the telescopic sum identity given in \eqref{eq:telescope}.  
We will establish error estimates for
\begin{equation}\label{eq:telebpe}
\sum_{l=0}^L \bar{E}_l^{N_l}(\varphi), \qquad \bar{E}_l^{N_l}(\varphi)=E_l^{N_{l}}(\varphi)-E_l(\varphi) \ ,
\end{equation}
where
\begin{align}\label{eq:el}
E_l^{N_{l}}(\varphi) = 
\frac{\frac{1}{N_l}\sum_{i=1}^{N_l} \varphi(\theta^i,x_{0:n}^i)H_{1,\theta^i}(\theta^i,z_{0:n}^i)}{\frac{1}{N_l}\sum_{i=1}^{N_l} H_{1,\theta^i}(\theta^i,z_{0:n}^i)} 
- \frac{\frac{1}{N_l}\sum_{i=1}^{N_l}\varphi(\theta^i,x_{0:n}'^i)H_{2,\theta^i}(\theta^i,z_{0:n}^i)}{\frac{1}{N_l}\sum_{i=1}^{N_l}H_{2,\theta^i}(\theta^i,z_{0:n}^i)}
\end{align}
is a consistent estimator of 
$E_l(\varphi) := 
\mathbb{E}_{\pi_{h_l}}[\varphi(\theta,X_{0:n})] - \mathbb{E}_{\pi_{h_{l-1}}}[\varphi(\theta,X_{0:n})]$. 
Therefore \eqref{eq:telebpe} is a consistent estimator of 
$\mathbb{E}_{\pi_{h_L}}[\varphi(\theta,X_{0:n})]$ and the the MSE \eqref{eq:mse}
can be bounded, up to a constant, by the sum of the squared error of \eqref{eq:telebpe} 
and Bias$(L,\varphi)^2$, as given by \eqref{eq:bias}, 
which is $\mathcal{O}(h_L)$ for example using Euler Maruyama.

Using $\mathbb{E}$ to denote the expectation w.r.t.~the law associated to our algorithm, assuming the Markov chain is started in stationarity, our objective is therefore to investigate
\begin{equation}
\mathbb{E}[(\sum_{l=0}^L \bar{E}_l^{N_l}(\varphi))^2] = \sum_{l=0}^L\mathbb{E}[\bar{E}_l^{N_l}(\varphi)^2] 
\label{eq:var_cmcmc}
\end{equation}
so as to optimally allocate $N_0,\dots,N_L$ as described in the introduction.
Thus we must investigate terms such as $\mathbb{E}[\bar{E}_l^{N_l}(\varphi)^2]$ for a given $l$.

\subsection{Analysis}

Below $\mathcal{P}(\mathsf{W})$ are the collection of probability measures on $(\mathsf{W},\mathcal{W})$.

\begin{hypA}
\label{hyp:A}
For every $y\in\mathsf{Y}$
there exist $0< \underline{C}< \overline{C}<+\infty$
such that for every
$x\in\mathsf{X}$, $\theta\in\Theta$, 
$$
\underline{C} \leq g_{\theta}(x,y) \leq \overline{C}.
$$
For every $y\in\mathsf{Y}$, $g_{\theta}(x,y)$
 is globally Lipschitz on $\mathsf{X}\times\Theta$.
\end{hypA}
\begin{hypA}
\label{hyp:B}
For any $0\leq k\leq n$, $q\in\{1,2\}$ there exists a $\beta > 0 $ 
such that for any 
$\varphi\in\mathcal{B}_b(\Theta\times\mathsf{X}^{k+1})\cap\textrm{Lip}(\Theta\times\mathsf{X}^{k+1})$
there exists a $C<+\infty$
$$
\left( \int_{\Theta\times\mathsf{X}^{2k+2}}|\varphi(\theta,x_{0:k})-\varphi(\theta,x_{0:k}')|^q 
\prod_{p=1}^k Q_{\theta,h,h'}(z_{k-1},z_k) \pi_{\theta}(\theta)\nu_\theta(z_0)d\theta dz_{0:k} \right)^{3-q}
\leq C ( h')^{\beta}.
$$
\end{hypA}
\begin{hypA}
\label{hyp:C}
Suppose that for any $n>0$ there exist a $\xi\in(0,1)$ and $\nu\in\mathcal{P}(\mathsf{W})$ such that for each $w\in\mathsf{W}$, $\varphi\in\mathcal{B}_{b}(\mathsf{W})\cap\textrm{Lip}(\mathsf{W})$, $h,h'$:
$$
\int_{\mathsf{W}} \varphi(w') K(w,dw')  \geq \xi \int_{\mathsf{W}} \varphi(w) \nu(dw).
$$
$K$ is $\eta$-reversible, that is, $\int_{w\in B}\eta(dw)K(w,A)=\int_{w\in A}\eta(dw)K(w,B)$ for any $A, B\in\mathcal{W}$. 
\end{hypA}

We note that (A\ref{hyp:A}) can be verified for some state-space models (especially if $\mathsf{Y}$ and $\Theta$ are compact) and (A\ref{hyp:C}) can be verified for a PMCMC
kernel, if $\Theta,\mathsf{X}$ are compact - indeed, the constants would all be independent of $n$ under appropriate settings of the algorithm.

\begin{theorem}\label{thm:main}
Assume (A\ref{hyp:A}-\ref{hyp:C}). Then for any $n>0$,
there exists a  $\beta > 0 $ 
such that for any $\varphi\in\mathcal{B}_b(\Theta\times\mathsf{X}^{n+1})\cap\textrm{Lip}(\Theta\times\mathsf{X}^{n+1})$
there exists a $C<+\infty$ such that
$$
\mathbb{E}\Bigg[\Bigg(
\frac{\frac{1}{N}\sum_{i=1}^N \varphi(\theta^i,x_{0:n}^i)H_{1,\theta^i}(\theta^i,z_{0:n}^i)}{\frac{1}{N}\sum_{i=1}^N H_{1,\theta^i}(\theta^i,z_{0:n}^i)} 
- \frac{\frac{1}{N}\sum_{i=1}^N\varphi(\theta^i,x_{0:n}'^i)H_{2,\theta^i}(\theta^i,z_{0:n}^i)}{\frac{1}{N}\sum_{i=1}^NH_{2,\theta^i}(\theta^i,z_{0:n}^i)}
$$
$$
-
\Bigg(
\frac{\mathbb{E}_{\pi_{h,h'}}[\varphi(\theta,X_{0:n})H_{1,\theta}(\theta,Z_{0:n})]}{\mathbb{E}_{\pi_{h,h'}}[H_{1,\theta}(\theta,Z_{0:n})]} - 
\frac{\mathbb{E}_{\pi_{h,h'}}[\varphi(\theta,X_{0:n}')H_{2,\theta}(\theta,Z_{0:n})]}{\mathbb{E}_{\pi_{h,h'}}[H_{2,\theta}(\theta,Z_{0:n})]}
\Bigg)
\Bigg)^2\Bigg] \leq \frac{C(h')^{\beta}}{N}.
$$
\end{theorem}

\begin{proof}
The result follows by using Lemma C.3.~of \cite{mlpf}, the $C_2-$inequality, the boundedness of certain quantities and 
Proposition \ref{prop:main}.
The proof is omitted as it is similar to the calculations in \cite{mlpf}.
\end{proof}

\subsection{A Return to Multilevel Considerations}


Returning to Section \ref{sec:mlmc}, we assume that $h_l = 2^{-l}$ and
introduce the further assumption
\begin{ass}\label{ass:cost}
The cost to simulate $E_l^{N_l}$ in \eqref{eq:el} is controlled by 
$\mathsf{C}(E_l^{N_l}) \leq C N_l h_l^{-\gamma}$, and the bias is controlled by
$$
|\mathbb{E}_{\pi_{h_L}}(\varphi(\theta,X_{0:n}))-\mathbb{E}_\pi(\varphi(\theta,X_{0:n}))| \leq C h_L^{\alpha} \ ,
$$
for $\gamma, \alpha, C > 0$.
\end{ass}
Following assumption (A\ref{hyp:B}), $\alpha=\beta/2$ satisfies the above, 
but it may be larger, e.g. for Euler-Maruyama in which $\alpha=\beta$.

Given $\epsilon > 0$, in order to ensure the MSE is $\mathcal{O}(\epsilon^2)$,  
the term \eqref{eq:bias} 
must be $\mathcal{O}(\epsilon^2)$.
Following from Assumption (A\ref{hyp:B}), it suffices to let 
$L\propto {2|\log(\epsilon)|}/{\beta}$ so that $h_L = \epsilon$.

Following from Theorem \ref{thm:main}, 
$$
\sum_{l=0}^L\mathbb{E}[\bar{E}_l^{N_l}(\varphi)^2]  \leq C\sum_{l=0}^L \frac{h_l^\beta}{N_l},
$$
and note that the constant $C$ may depend upon the time parameter $n$, 
which has been suppressed from the notation; we return to this point below.

Suppose we minimize COST $=\sum_{l=0}^L h_l^{-\gamma} N_l$ subject to $\sum_{l=0}^L \frac{h_l^\beta}{N_l}=\mathcal{O}(\epsilon^2)$ as a function of $N_0,\dots,N_L$.
This is exactly considered in \cite{giles:08} for $\gamma=1$ and
later in \cite{scheichl} for $\gamma\neq 1$, and yields that 
\begin{equation}\label{eq:en}
N_l\propto \varepsilon^{-2}K_L h_l^{(\beta+\gamma)/2},
\end{equation}
where $K_L = \sum_{l=1}^{L} h_l^{(\beta-\gamma)/2}$ (see also \cite{mlpf, mlsmcnorm}).
This gives a cost of $\mathcal{O}(\varepsilon^{-2} K_L^{2})$ 
per time step. Hence the following corollary is immediate.

\begin{cor}[ML Cost]\label{cor:cost} 
Given (A\ref{hyp:A}-\ref{hyp:C}) and Assumption \ref{ass:cost}, for any $n>0$ and any
$\varphi\in\mathcal{B}_b(\Theta\times\mathsf{X}^{n+1})\cap\textrm{Lip}(\Theta\times\mathsf{X}^{n+1})$,
$(L,\{N_l\}_{l=1}^L)$ can be chosen such that the estimator
$
\sum_{l=1}^L E_l^{N_l}(\varphi),
$ 
with $E_l^{N_l}$ given in \eqref{eq:el}, satisfies
$$
\mathbb{E}\left [ |\sum_{l=1}^L E_l^{N_l}(\varphi) - \mathbb{E}_\pi(\varphi(\theta,X_{0:n}))|^2 \right ] \leq C \epsilon^2 \ ,
$$
for some $C>0$, for a total cost controlled by
\begin{equation}\label{eq:totcost}
{\rm COST} \leq C 
\begin{cases}
\epsilon^{-2}, & \text{if} \quad \beta > \gamma,\\
\epsilon^{-2} |\log(\epsilon)|^2, & \text{if} \quad \beta = \gamma,\\
\epsilon^{- \left( 2 + \frac{\gamma-\beta}{\alpha} \right)}, & \text{if} \quad \beta < \gamma.
\end{cases}
\end{equation}
\end{cor}

In contrast, for the same scenario, the computational cost of PMCMC is 
$\mathcal{O}(\epsilon^{-2-\gamma/\alpha})$ per time step, which 
is asymptotically greater than the method developed here.

It is remarked that all of our constants depend upon the time parameter (number of data points) and this element has been ignored. This is due to the technical complexity of the approach. 
We expect that the constants can be made time-uniform, 
and hence we conjecture that the results hold true uniformly in time. 
Then $N_l$ can be chosen as above, 
and for Euler Maruyama ($\beta=\gamma=1$ \cite{GrahamTalay})
the cost for a given $n$ will be 
$\mathcal{O}(n^2|\log(\epsilon)|^2\epsilon^{-2})$, with similar results for $\beta\neq 1$,
according to \eqref{eq:totcost}. 
This results because one needs to take $M=\mathcal{O}(n)$ for the particle filter in PMMH 
\cite{andrieu} and the cost to obtain a 
single sample particle filter trajectory is $\mathcal{O}(n)$.
A verification of this  is left for future work.

\section{Numerical Simulations}\label{sec:numerics}

\subsection{Ornstein-Uhlenbeck process}

First, we consider the following Ornstein-Uhlenbeck process,
\begin{align*}
  dX_t &= \theta(\mu - X_t) + \sigma dW_t,\qquad X_0 = x_0 \\
  Y_k | X_{k\delta} &\sim \mathcal{N}(X_{k\delta}, \tau^2),
\end{align*}
where $\mathcal{N}(m, \tau^2)$ denotes the Normal distribution with mean $m$
and variance $\tau^2$. Further, the parameters $(\theta,\sigma)$ are unknown
and are given the following priors,
\begin{equation*}
  \theta \sim \mathcal{G}(1, 1), \qquad
  \sigma \sim \mathcal{G}(1, 0.5)
\end{equation*}
where $\mathcal{G}(a,b)$ denotes the Gamma distribution with shape $a$ 
and scale $b$. 
The remaining parameters are defined as constants, $x_0 =
0$, $\mu = 0$, $\delta = 0.5$, and $\tau^2 = 0.2$. A data set with 100
observations is simulated with $\theta = 1$ and $\sigma = 0.5$.

\subsection{Langevin SDE}

Consider the following Langevin SDE,
\begin{align*}
  dX_t &= \frac{1}{2}\nabla\log\pi(X_t) + \sigma dW_t,\qquad X_0 = x_0\\
  Y_k|X_{k\delta} &\sim \mathcal{N}(0, \tau^2\exp{X_{k\delta}}),
\end{align*}
where $\pi(x)$ denote the probability density function of a Student's
$t$-distribution with $\theta$ degrees of freedom. The parameters of interest
are $(\theta,\sigma)$, and these are given prior,
\begin{equation*}
  \theta \sim \mathcal{G}(1, 1), \qquad
  \sigma \sim \mathcal{G}(1, 1)
\end{equation*}
The constants are $x_0 = 0$ and $\delta = 1$. A data set with 1,000
observations is simulated with $\theta = 10$, $\sigma = 1$, and $\tau^2 = 1$.

\subsection{Simulation settings}

The simulations proceed as the following. Let $h = 2^{-l}$ be the accuracy
parameter. At each level $l$, we set the number of particles in the PMCMC
kernel be $M = \mathcal{O}(n)$ fixed, and set the number of PMCMC samples for
estimation according to the multilevel analysis. Let $N_l^L$ denote the number
of samples at level $l$ within a simulation that targets $L$-level error, $L =
1,\dots$. The value of $N_0^1$ is determined empirically with variance
estimated from 100 samplers. For comparison, a single-level PMCMC sampler is
also considered for each $L$. Its number of samples $N^L$ is determined
empirically by running 100 simulations simultaneously. And these chains are run
until the estimated error of the 100 estimates matches that of the multilevel
sampler.  In all situations, a fixed burn-in period of 10,000 iterations is
used.  This is reasonable given the fast decorrelation of the chains, as
illustrated by the estimated autocorrelation of the single level PMCMC sampler
for $L=8$ in Figure \ref{fig:acf}. The autocorrelation functions look similar
for all $l$ for the multilevel sampler.

\subsection{Results}

We consider the choice of $M = \mathcal{O}(n)$.
The main results of the cost vs.~error are shown
in Figure~\ref{fig:cost}. The estimated cost rates are listed in
table~\ref{tab:rates}.  
It is shown in the appendix that for Euler discretization the method satisfies the 
assumptions (A\ref{hyp:A}-\ref{hyp:C}) with $\beta=2$ in (A\ref{hyp:B}), 
since the diffusion term $b_\theta$ is constant in $x$ \cite{GrahamTalay}.
Furthermore, Assumption \ref{ass:cost} holds with $\gamma=\alpha=1$.  
Therefore, the theoretical results of Theorem \ref{thm:main} and Corollary \ref{cor:cost}
predict the rate $\mathcal{O}(\epsilon^{-2})$.
Standard PMMH will incur a cost of $\mathcal{O}(\epsilon^{-3})$.
The numerical results confirm this.  

\begin{figure}
  \includegraphics[width=\linewidth]{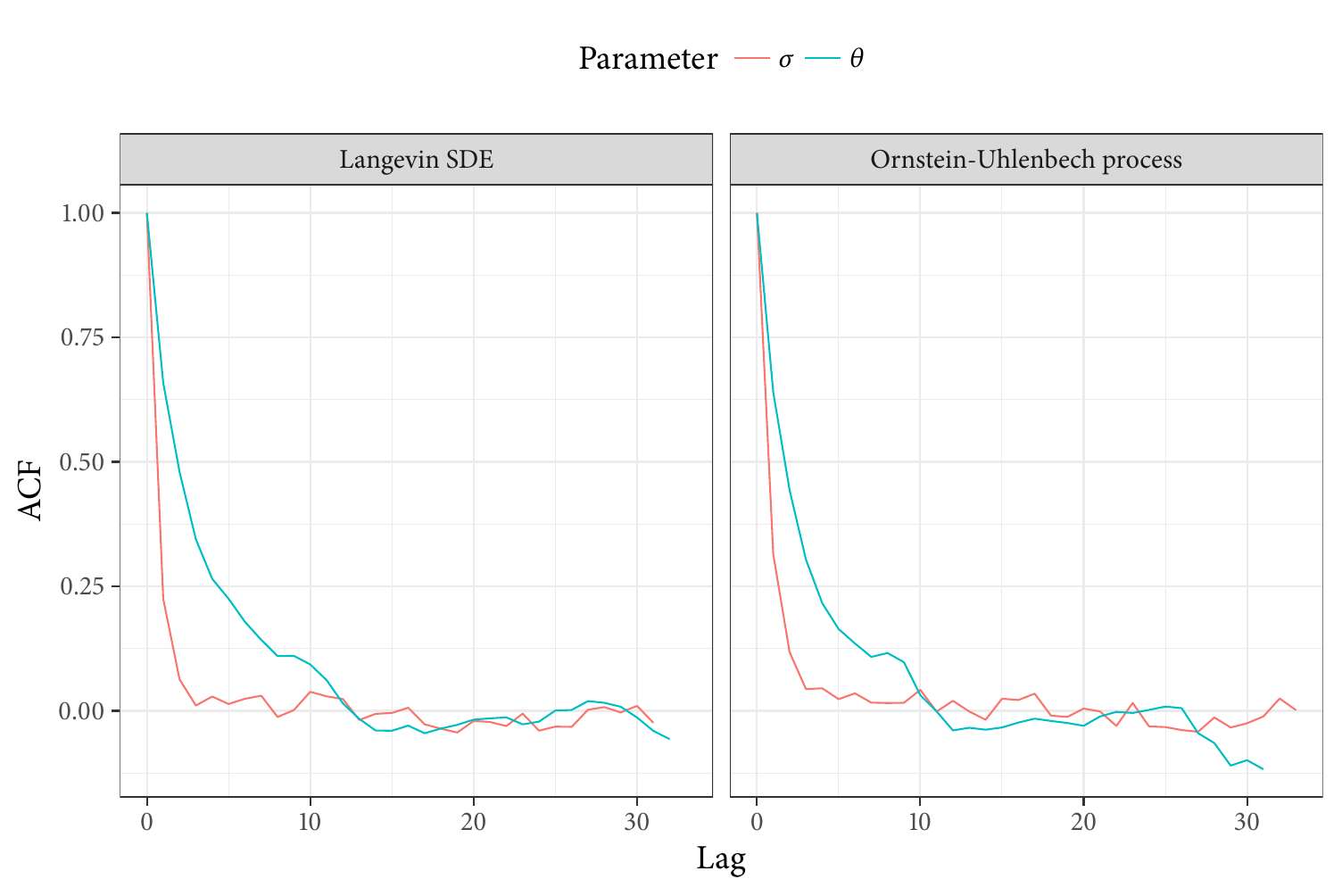}
  \caption{Autcorrelation of typical PMCMC chains.}
  \label{fig:acf}
\end{figure}

\begin{figure}
  \includegraphics[width=\linewidth]{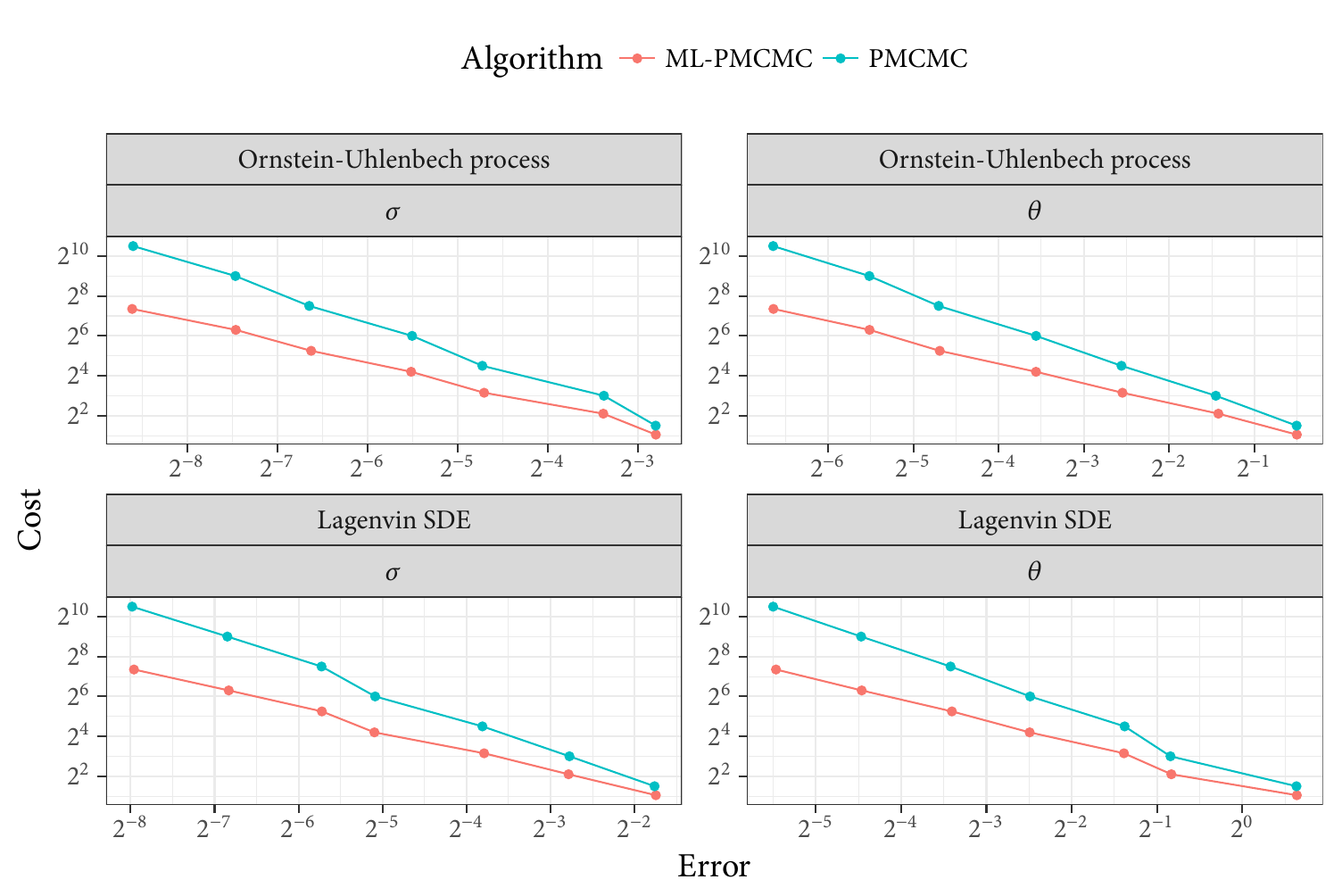}
  \caption{Cost vs. MSE for the 2 parameters for each of the 2 SDEs.}
  \label{fig:cost}
\end{figure}

\begin{table}\centering
  \begin{tabular}{|c|c|c|c|}
\hline
    Model & Parameter & ML-PMCMC & PMCMC \\
    \hline
    Ornstein-Uhlenbech process & $\theta$ & $-1.022$ & $-1.463$ \\
                               & $\sigma$ & $-1.065$ & $-1.522$ \\
    Langevin SDE               & $\theta$ & $-1.060$ & $-1.508$ \\
                               & $\sigma$ & $-1.023$ & $-1.481$ \\
    \hline
  \end{tabular}
\caption{Estimated rates of convergence of MSE with respect to cost 
for various parameters, fitted to the curves in Figure \ref{fig:cost}.}
\label{tab:rates}\end{table}

\subsubsection*{Acknowledgements}
AJ \& YZ were supported by an AcRF tier 2 grant: R-155-000-161-112. AJ is affiliated with the Risk Management Institute, the Center for Quantitative Finance 
and the OR \& Analytics cluster at NUS. KK \& AJ acknowledge CREST, JST for additionally supporting the research.
KJHL was supported by ORNL LDRD Strategic Hire grant 32112580.

\appendix

\section{Technical Results}

A Markov kernel $K$ can be viewed as a linear operator $(Kf)(w)=\int K(w,dw^*)f(w^*)$ for $f:\mathsf{W}\rightarrow\mathbb{R}$
on a Hilbert space
$$
L^2_0(\eta) := \{f : \mathsf{W} \rightarrow \mathbb{R}; \int|f(w)|\eta(dw)<\infty,
\int f(w)\eta(dw) = 0\}
$$
with an inner product $\langle f,g\rangle=\int f(w)g(w)\eta(dw)$ and 
norm $\|f\|_2=\sqrt{\langle f,f\rangle}$. Let $\|K\|_2=\sup_{f\in L^2_0(\eta),f\neq 0}\|Kf\|_2/\|f\|_2$ be the operator norm.

By D\"oblin condition (A\ref{hyp:C}), we have the total variation distance bound $\|K(w,\cdot)-\eta\|_{\mathrm{TV}}=\sup_{A\in \mathcal{W}}|K(w,A)-\eta(A)|\le 1-\xi\ (\forall w\in \mathsf{W})$
for some $\xi\in (0,1)$. Since $K$ is an Metropolis-Hastings kernel, it has $\eta$-reversibility. Therefore, the total variation bound implies $L^2$-spectral gap 
$$
\|K^m\|_2\le (1-\xi)^m \ ,
$$
by Theorem 2.1 of \cite{MR1448322}. 

For $\mu$ a finite measure on a measurable space $(\mathsf{E},\mathcal{E})$ and $\varphi\in\mathcal{B}_b(\mathsf{E})$
$$
\mu(\varphi) = \int_E \varphi(x)\mu(dx).
$$

%
%

Defining $v^i=(\theta^i,z_{0:n}^i)$ as the relevant variables of $w^i$ from the MCMC kernel, 
and defining 
$$
\tilde{\varphi}_h(v^i) := \Big\{\varphi(\theta^i,x_{0:n}^i)H_{1,\theta^i}(\theta^i,z_{0:n}^i) - 
\varphi(\theta^i,x_{0:n}'^i)H_{2,\theta^i}(\theta^i,z_{0:n}^i)\Big\},
$$
we are interested in estimates of the form:
$$
\frac{1}{N}\sum_{i=1}^N\tilde{\varphi}_h(v^i). 
$$

\begin{proposition}\label{prop:main}
Assume (A\ref{hyp:A}-\ref{hyp:C}). 
Suppose that $\{W^i\}_i$ is a Markov chain with the Markov kernel $K$, and $W^1\sim \eta$. 
Then for any $n>0$,
there exists a  $\beta > 0 $ 
such that for any $\varphi\in\mathcal{B}_b(\Theta\times\mathsf{X}^{n+1})\cap\textrm{Lip}(\Theta\times\mathsf{X}^{n+1})$
there exists a $C<+\infty$ such that
$$
\mathbb{E}\Big[\Big(\frac{1}{N}\sum_{i=1}^N\tilde{\varphi}_h(V^i)-\pi_{h,h'}(\tilde{\varphi}_h)\Big)^2\Big] \leq \frac{C(h')^{\beta}}{N}, 
$$
where $V^i=(\theta, Z_{0:n}^i)$ is the relevant variables of $W^i$. 
\end{proposition}

\begin{proof}
Denote the map $w^i\mapsto v^i$ by $\psi$. Then 
$$
\mathbb{E}\Big[\Big(\frac{1}{N}\sum_{i=1}^N\tilde{\varphi}_h(V^i)-\pi_{h,h'}(\tilde{\varphi}_h)\Big)^2\Big] =
\mathbb{E}\Big[\Big(\frac{1}{N}\sum_{i=1}^Nf(W^i)\Big)^2\Big]
$$
for $f(w)=\tilde{\varphi}_h\circ\psi(w)-\eta(\tilde{\varphi}_h\circ\psi)=\tilde{\varphi}_h(v)-\pi_{h,h'}(\tilde{\varphi}_h)$. 
By simple algebra, 
\begin{align*}
\mathbb{E}\Big[\Big(\frac{1}{N}\sum_{i=1}^Nf(W^i)\Big)^2\Big]=&
\frac{1}{N^2}\sum_{i,j=1}^N\langle f,K^{|i-j|} f\rangle\\
=&\frac{1}{N}\|f\|_2^2 +\frac{2}{N^2}\sum_{n=1}^{N-1}(N-n)\langle f,K^nf\rangle \\
\le &\frac{1}{N}\|f\|_2^2 +\frac{2}{N^2}\sum_{n=1}^{N-1}(N-n)\|K\|_2^n\|f\|_2^2 \\
\le &\frac{1}{N}\|f\|_2^2 +\frac{2}{N}\sum_{n=1}^\infty\|K\|_2^n\|f\|_2^2 =\frac{1}{N}\frac{3-\|K\|_2}{1-\|K\|_2}\|f\|_2^2\le \frac{3\|f\|_2^2}{N\xi}. 
\end{align*}
On the other hand, by Lemma \ref{lem:tres1}, 
\begin{align*}
	\|f\|_2^2=\mathbb{E}_{\pi_{h,h'}}\left[\left\{\tilde{\varphi}_h(V^i)-\pi_{h,h'}(\tilde{\varphi}_h)\right\}^2\right]\le C(h')^\beta. 
\end{align*}
Thus, the claim follows. 
\end{proof}


\begin{lemma}\label{lem:tres1}
Assume (A\ref{hyp:A}-\ref{hyp:B}).  Then for any $n>0$, $q\in\{1,2\}$  there exists a  $\beta > 0 $ 
such that for any $\varphi\in\mathcal{B}_b(\Theta\times\mathsf{X}^{n+1})\cap\textrm{Lip}(\Theta\times\mathsf{X}^{n+1})$
there exists a $C<+\infty$
$$
\left(\mathbb{E}_{\pi_{h,h'}}[|\varphi(\theta,X_{0:n})H_{1,\theta}(\theta,Z_{0:n})-\varphi(\theta,X_{0:n}')H_{2,\theta}(\theta,Z_{0:n})|^q] \right)^{3-q} \leq C (h')^{\beta}.
$$
\end{lemma}

\begin{proof}
We prove the result for $q=1$, the case $q=2$ being almost the same.
The result is proved by induction on $n$. Set $n=1$, then
$$
\mathbb{E}_{\pi_{h,h'}}[|\varphi(\theta,X_{0:1})H_{1,\theta}(\theta,Z_{0:1})-
\varphi(\theta,X_{0:1}')H_{2,\theta}(\theta,Z_{0:1})|] 
= 
$$
$$
\Big(\int_{\theta\times \mathsf{Z}^2} |\varphi(\theta,x_{0:1})H_{1,\theta}(\theta,z_{0:1})-\varphi(\theta,x_{0:1}')H_{2,\theta}(\theta,z_{0:1})| G_{1,\theta}(z_1)
\nu_\theta(z_0) \pi_{\theta}(\theta) d(z_{0:1},\theta)\Big)\times
$$
$$
\Big(\int_{\theta\times \mathsf{Z}^2}G_{1,\theta}(z_1)
\nu_\theta(z_0)\pi_{\theta}(\theta) d(z_{0:1},\theta)\Big)^{-1}.
$$
As $G_{1,\theta}(z)$ is uniformly (in $\theta,z$) bounded below, the denominator on the R.H.S.~is uniformly lower bounded by a constant that is independent of $h,h'$.
The numerator on the R.H.S.~is 
$$
\int_{\theta\times \mathsf{Z}^2} |\varphi(\theta,x_{0:1})H_{1,\theta}(\theta,z_{0:1})-\varphi(\theta,z_{0:1}')H_{2,\theta}(\theta,z_{0:1})| G_{1,\theta}(z_1)
\nu_\theta(z_0) \pi_{\theta}(\theta) d(z_{0:1},\theta) = 
$$
$$
\int_{\theta\times \mathsf{Z}^2}
|\varphi(\theta,x_{0:1})g_{\theta}(x_1,y_1)-\varphi(\theta,x_{0:1}')g_{\theta}(x_1',y_1)|
\nu_\theta(z_0) \pi_{\theta}(\theta) d(z_{0:1},\theta)
$$
Application of (A\ref{hyp:B}) hence yields
$$
\mathbb{E}_{\pi_{h,h'}}[|\varphi(\theta,X_{0:1})H_{1,\theta}(\theta,Z_{0:1})
-\varphi(\theta,X_{0:1}')H_{2,\theta}(\theta,Z_{0:1})|] \leq C (h')^{\beta/2}.
$$
Assuming the result for $k-1$, $k> 1$, by the above argument we only have to consider
$$
\int_{\theta\times \mathsf{Z}^{k+1}} |\varphi(\theta,x_{0:k})H_{1,\theta}(\theta,z_{0:k})-\varphi(\theta,x_{0:k}')H_{2,\theta}(\theta,z_{0;k})| \nu_\theta(z_0)
\prod_{p=1}^k
G_{p,\theta}(z_p)
Q_{\theta,h,h'}(z_{p-1},z_p) \pi_{\theta}(\theta) d(z_{0:k},\theta) = 
$$
$$
\int_{\theta\times \mathsf{Z}^{k+1}} |\varphi(\theta,x_{0:k})
\prod_{p=1}^k g_{\theta}(x_p,y_p)
-\varphi(\theta,x_{0:k}')\prod_{p=1}^k g_{\theta}(x_p',y_p)| 
\prod_{p=1}^k
Q_{\theta,h,h'}(z_{p-1},z_p) \nu_\theta(z_0)\pi_{\theta}(\theta) d(z_{0:k},\theta).
$$
The R.H.S.~can be upper-bounded by
$$
\int_{\theta\times \mathsf{Z}^{k+1}} \varphi(\theta,x_{0:k})g_{\theta}(x_k,y_k)|
\prod_{p=1}^{k-1}g_{\theta}(x_p,y_p)
-\prod_{p=1}^{k-1} g_{\theta}(x_p',y_p)|
\prod_{p=1}^k Q_{\theta,h,h'}(z_{p-1},z_p) \nu_\theta(z_0) \pi_{\theta}(\theta) d(z_{0:k},\theta)
+
$$
$$
\int_{\theta\times \mathsf{Z}^{k+1}}\prod_{p=1}^{k-1} g_{\theta}(x_p',y_p)
|\varphi(\theta,x_{0:k})g_{\theta}(x_k,y_k)-\varphi(\theta,x_{0:k}')g_{\theta}(x_k',y_k)|
\prod_{p=1}^k
Q_{\theta,h,h'}(z_{p-1},z_p) \nu_\theta(z_0) \pi_{\theta}(\theta) d(z_{0:k},\theta).
$$
The first term can be treated by the induction hypothesis and the second term via 
(A\ref{hyp:B}) which completes the proof.
\end{proof}

%
%

\section{Coupling Euler Approximations}\label{app:couple_euler}

Consider $(x,x')\in\mathsf{X}^2$, the current position of the discretized diffusions. Now we have $h,h'$ the discretization levels,
with $0<h<h'$ and for simplicity set $h'=2h$. Associated to the discretization level $h$ (resp.~$h'$), one must sample $k=\delta/h$ (resp.~$k'=\delta/h'$) points to obtain the
sampled position of the diffusion at the next observation time. Set $X(0)=X'(0)\sim f_\theta(x)d x$ then one can sample the fine discretization, for $m\in\{0,\dots,k-1\}$ as
$$
X(m+1) = X(m) + h a_\theta(X(m)) + \sqrt{h} b_\theta(X(m))\xi(m)
$$
where $\xi(m)\stackrel{\textrm{i.i.d.}}{\sim}\mathcal{N}(0,I_d)$ ($I_d$ is the $d\times d$ identity matrix). For the course discretization, using the same
simulated $\xi(0),\dots,\xi(k-1)$ we set for $m\in\{0,\dots,k'-1\}$
$$
X'(m+1) = X'(m) + 2h a_\theta(X'(m)) + \sqrt{h} b_\theta(X'(m))[\xi(2m)+\xi(2m+1)].
$$

Now, we want to check conditions (A\ref{hyp:B}) and (A\ref{hyp:C}) under Assumption \ref{ass:diff} (i,ii), (A\ref{hyp:A}) and the following assumption. 
\begin{ass}\label{ass:compact}
$\Theta$ is a compact set of $\mathbb{R}^{d_\theta}$, and 
$\pi_\theta:\Theta\rightarrow\mathbb{R}_+$ and 
$q(\theta^*|\theta):\Theta^2\rightarrow\mathbb{R}_+$
are continous and strictly positive.	
\end{ass}

By assumption, 
$Q_{\theta,h,h'}(z,z^*)$ is the density of $Z^*=(X(k),X'(k'))$ given $Z=(X(0),X'(0))$. 
Then, under Assumption \ref{ass:diff} (i, ii), 
the condition (A\ref{hyp:B}) is satisfied with $\beta =1$ for any $q=1,2$,
since this is the $L^q$ bound of the Euler-Maruyama scheme 
(in fact for constant diffusion coefficient $b_\theta$
it coincides with the Milstein method and $\beta=2$) \cite{GrahamTalay}. 

Next, we want to check the condition (A\ref{hyp:C}). The proposal density $\psi$ on $\mathsf{W}=\Theta\times \{1,\ldots, M\}^n\times Z^{M(n+1)}\times \{1,\ldots, M\}$ of PMMH is 
$$
\psi(w, w^*)=p(a_{0:n-1}^{*,1:M}, z_{0:n}^{*,1:M}|\theta^*)q(\theta^*|\theta)\frac{G_{n,\theta^*}(z_n^{*,l})}{\sum_{i=1}^MG_{n,\theta^*}(z_n^{*,i})}. 
$$
where $w=(\theta,a_{0:(n-1)}^{1:M}, z_{1:n}^{1:M},k)$, $w^*=(\theta^*,a_{0:(n-1)}^{*,1:M}, z_{1:n}^{*,1:M},l)$,
and $p(a_{0:n-1}^{*,1:M}, z_{0:n}^{*,1:M}|\theta^*)$ is defined in \eqref{eq:pf_law}. 
The transition kernel $K$ is
$$
K(w,dw^*)=\psi(w,w^*)\alpha(w,w^*)dw^* + \delta_w(dw^*)R(w), 
$$
where the acceptance probability $\alpha(w,w^*)$ is 
$$
\alpha(w,w^*)=\min\left\{1,\frac{p_{h,h'}^M(y_{0:n}|\theta^*)\pi_\theta(\theta^*)q(\theta|\theta^*)}{p_{h,h'}^M(y_{0:n}|\theta)\pi_\theta(\theta)q(\theta^*|\theta)}\right\}, 
$$
and the rejection probability $R(w)$ is 
$$
R(w)=1-\int_{w^*} \psi(w,w^*)\alpha(w,w^*)dw^*. 
$$
By (A\ref{hyp:A}) together with Assumption \ref{ass:compact}, 
$C_1=\inf_{w\in W}\alpha(w,w^*)>0$, and 
$$
\inf_w \psi(w,w^*)\ge \left\{\min_{\theta,\theta^*}q(\theta^*|\theta)\right\}p(a_{0:n-1}^{*,1:M}, z_{0:n}^{*,1:M}|\theta^*)\frac{G_{n,\theta^*}(z_n^{*,l})}{\sum_{i=1}^MG_{n,\theta^*}(z_n^{*,i})}=:C_2\psi(w^*)
$$
for a constant $C_2=\min_{\theta,\theta^*}q(\theta^*|\theta)>0$
with a probability density $\psi(w^*)$. 
Thus, we have
\begin{align*}
	K(w,dw^*)\ge C_1C_2 \psi(w^*)dw^*
\end{align*}
In particular, the condition (A3) is satisfied.

\end{document}